\documentclass[letterpaper, 10 pt, conference]{IEEEtran}

\IEEEoverridecommandlockouts

\usepackage{amsmath}
\usepackage{graphicx}
\usepackage{amssymb}
\usepackage{amsthm}
\usepackage{bbm}
\usepackage{easylist}
\usepackage{fp}
\usepackage{siunitx}
\usepackage{tikz}
\usepackage{xcolor}
\usepackage{color}

\usepackage[ruled]{algorithm2e}
\usepackage[lofdepth,lotdepth]{subfig}

\newtheorem{remark}{Remark}
\newtheorem{theorem}{Theorem}

\newtheorem{problem}{Problem}

\newtheorem{corollary}{Corollary}

\newtheorem{model}{Model}

\ifCLASSINFOpdf
\else
\fi
\usepackage{url}

\begin{document}
%
\title{Multi-Fleet Platoon Matching:\\ A Game-Theoretic Approach  }

\author{\IEEEauthorblockN{Alexander Johansson, Ehsan Nekouei, Karl Henrik Johansson, Jonas M\aa rtensson}
\thanks{*
		This work is supported by the Strategic Vehicle Research and Innovation Programme, Horizon2020, the Knut and Alice Wallenberg Foundation, the Swedish Foundation for Strategic Research and the Swedish Research Council.}
 \thanks{A. Johansson, E. Nekouei, K. H. Johansson and J. M\aa rtensson are with the Integrated Transport Research Lab and Department of Automatic Control,
School of Electrical Engineering and Computer Science, KTH Royal Institute
of Technology, Stockholm, Sweden.,
 SE-100 44 Stockholm, Sweden. Emails:
         {\tt\small \{alexjoha, nekouei, kallej, jonas1\}@kth.se}}}



%


\maketitle

\thispagestyle{plain}
\pagestyle{plain}

\begin{abstract}

We consider the platoon matching problem for a set of  trucks with the same origin, but different destinations. It is assumed that the vehicles benefit from traveling in a platoon for instance through reduced fuel consumption. The vehicles belong to different fleet owners and their strategic interaction is modeled as a non-cooperative game where the vehicle actions are their departure times. Each truck has a preferred departure time and its  utility function is defined as the difference between its benefit from platooning and the cost of deviating from its preferred departure time. We show that the platoon matching game is an exact potential game. An algorithm based on best response dynamics is proposed for finding a Nash equilibrium of the game. At a Nash equilibrium, vehicles with the same departure time are matched to form a platoon. Finally, the total fuel reduction at the Nash equilibrium is studied and compared with that of a cooperative matching solution where a common utility function for all vehicles is optimized.

\end{abstract}

\IEEEpeerreviewmaketitle

\section{Introduction}

\subsection{Motivation}

 The transportation sector emitted $24 \%$ of the total $\text{CO}_2$ emissions due to fuel combustion in 2015 \cite{Publishing2017} and faces an enormous challenge of reducing its fuel consumption. Truck platooning has received attention due to its potential to significantly reduce the fuel consumption \cite{Tsugawa2016} \cite{Alam2015}. The potential fuel reduction by platooning is shown by numerical studies in \cite{Davila2013}, \cite{Bishop2017} and by field experiments in \cite{Alam2010}, \cite{Browand2004}, \cite{Tsugawa2016}. Truck platooning has other benefits, besides reduction in fuel consumption, e.g., decreasing the workload of the drivers, improving safety by reducing the human factor and reducing traffic congestion. The interested reader is referred to \cite{Alam2015} for a high-level introduction to truck platooning.

Platoon matching is an important step in platoon formation where a group of vehicles is divided in to smaller groups and each group will form a platoon. A review on planning strategies for truck platooning, including platoon matching, is given in \cite{Bhoopalam2018}. One solution for the platoon matching problem is to form platoons by solving a centralized optimization problem. When trucks belong to different fleet owners, they may instead seek individually to maximize their profit from platooning. This scenario can be modeled using non-cooperative game theory.

\subsection{Related Work}

Centralized solutions for the platoon matching problem have been proposed in the literature, e.g., \cite{Liang2016}, \cite{Larsson2015}, \cite{Hoef2018}, \cite{Boysen2018}. Centralized solutions rely on the assumption that the vehicles have a common objective function. This assumption is realistic when vehicles are owned by one transportation company, \emph{i.e.,} a single fleet.

The authors in \cite{Liang2016} studied the problem of determining optimal speed profiles, in terms of total fuel consumption, of two vehicles that are merging on the road and then platooning. The results are used to determine if two trucks save fuel or not by platooning.

The authors of \cite{Larsson2015} and \cite{Hoef2018} studied the platoon matching problem to minimize the total fuel consumption for vehicles in a large road network. In these works, it was assumed that vehicles are cooperative and share all information about their positions and missions with a centralized optimizer. In \cite{Larsson2015}, the optimal platoon matching problem was solved by formalizing a route optimization problem where the cost of traversing a road decreases with the number of vehicles platooning on it. In \cite{Hoef2018}, the matching problem was solved by first determining feasible platoon partners, in terms of mission dead-lines, and then finding the optimal platoon leaders such that the total fuel saving is maximized.

The impact of the time constraints of the vehicles on platooning was studied in  \cite{Boysen2018}. The authors considered a group of vehicles with the same origin and destination. They posed the platoon matching problem as a centralized optimization problem in which a group of trucks platoon if they depart from the origin simultaneously. Different from this work, we model the interaction among the vehicles as a non-cooperative game.

The authors of \cite{Farokhi2013} modeled the strategic interaction among cars and trucks as a non-cooperative game. The vehicles share origin and destination. Each car or truck minimizes its expected traveling cost by deciding on its departure time and trucks benefit from platooning. Different from \cite{Farokhi2013}, we allow the vehicles to have different destinations in a road network defined on a graph with tree topology. Hence, our models captures that platoons may split into sub-platoons along the routes.

\subsection{Main Contributions}

In this paper we consider the platoon matching problem for a set of trucks with competitive behavior. The vehicles have the same origin, but different destinations. The interaction among the vehicles is modeled as a non-cooperative game where the vehicles' actions are their departure times from a common origin. The utility function of each vehicle in the game is a combination of the monetary saving from platooning and a penalty for deviating from its preferred departure time. We show that the platoon matching game is a potential game and it admits at least one Nash equilibrium. An algorithm for finding a Nash equilibrium is proposed.  The algorithm is used to numerically study the equilibrium actions of the vehicles and their impacts on the total fuel saving due to platooning.

This paper is structured as follows. In Section \ref{problemform}, the model of the platoon matching problem is presented and the platoon matching game is defined. In Section \ref{DG}, we show that the platoon matching game is a potential game and an algorithm for finding a Nash equilibrium is proposed. In section \ref{simulation}, the algorithm is used in numerical experiments to study the  platoon matching problem when the preferred departure times are random. Finally, conclusions and future work are given in Section \ref{conclusions}.

\section{The Platoon Matching Problem} \label{problemform}

In this section, the platoon matching problem is formulated as a non-cooperative game. We start by describing the platoon matching problem.

\begin{problem}(Platoon Matching Problem)\label{problem1}
Consider a set of vehicles located at a node in a road network. The vehicles have different destinations and different preferred departure times. However, each vehicle can adjust its departure time to the departure time of another vehicle in order to benefit from platooning. The vehicles have different fleet owners, \emph{i.e.,} they seek to  maximize their own individual utility functions.  
\end{problem}

The set of vehicles is denoted by $\mathcal N=\{1,...,N\}$. The road network is represented by a directed graph $G~=~(V,E)$, where the vertices in $V$ represents the origin and destinations of the vehicles and the edges in $E$ represent the road segments. We assume that $G$ is a directed tree. The origin of vehicles is denoted by the node $v_1$, \emph{i.e.,} all vehicles are initially located at $v_1$. We assume that $v_1$ is the root node of $G$ and its degree is equal to $1$, see Fig. \ref{fig:RN} on page \pageref{fig:RN} for an example of $G$. The length of a road segment $e\in E$ is denoted by $d(e)$. Let $P(M) \subseteq E$ denote the road segments that at least one vehicle in the set $M\subseteq \mathcal N$ traverses and let $n(e,M)$ be the number of vehicles in the set $M$ that traverse the edge $e \subseteq E $.

\subsection{Vehicle Actions} \label{static}

Each vehicle $i\in\mathcal{N}$ has its own preferred departure time from the node $v_1$, but is allowed to deviate from its preferred departure time in order to platoon with other vehicles. The preferred departure time of vehicle $i$ is denoted by $t_i, \ i\in~ \mathcal{N}$. The vehicles have time windows in which they have to depart from the node $v_1$ due to the driver rest time regulations, deadlines etc. The time window of each vehicle $i \in \mathcal{N}$ is denoted by $T_i =[  \underline t_i , \bar t_i]$.

The actions of the vehicles are their departure times. Let $s_i$ denote the action of the vehicle $i\ \in \mathcal{N}$. Let $s =(s_1,\cdots,s_{N} )$ denote the strategy profile and let $s_{-i} =(s_1,\cdots,s_{i-1},s_{i+1},\cdots,s_{N} )$ denote the actions of all vehicles except vehicle $i$. The action of each vehicle $i$ is restricted to the set $S_i = T_i \cap \mathcal T$, where $\mathcal T =\{t_i|i \in \mathcal{N}\}$, \emph{i.e.,}  $S_i$ is the set of preferred departure times which lie within the time window of vehicle $i$. In Fig. \ref{fig:timewind}, the construction of the set $S_i$ is illustrated by an example. In this figure, the filled nodes indicate the preferred departure times of the vehicles $1,...,5$. The time window of vehicle $3$ is marked in the figure. The feasible actions of vehicle $3$ is then $S_3=\{ t_2, t_3, t_4 \}$, since the preferred departure times of the vehicles $2,3,4$ are included in the time window $T_3=[ \underline t_3 , \bar t_3]$.

\begin{figure}  
\centering
\begin{tikzpicture}[scale=0.8]

 \filldraw 
(0,0) circle (2pt) node[align=left,   below] {$t_1$} --
(2,0) circle (2pt) node[align=center, below] {$t_2$} -- 
(4,0) circle (2pt) node[align=right,  below] {$t_3$} --
(6,0) circle (2pt) node[align=right,  below] {$t_4$} --
(8,0) circle (2pt) node[align=right,  below] {$t_5$} ;
\draw [->]   (8,0) -- (10,0); 
\path (1,1.3) node(x1) {} 
      (1,-1) node(y1) {$\underline t_3$}
      (6.3,1.3) node(x2) {} 
      (6.3,-1) node(y2) {$\bar t_3$}
        (0.9,1) node(x3) {} 
      (6.4,1) node(y3) {}
      (3.65,-1.5) node(x4) {$S_3$};
\draw (x1) -- (y1);
\draw (x2) -- (y2);
\draw [<->] (x3) -- (y3) node [midway, fill=white] {$T_3$};
\draw (x4) [->] -- (2.2,-0.15);
\draw (x4) [->] -- (3.9,-0.6);
\draw (x4) [->] -- (5.7,-0.3);
\end{tikzpicture}
\caption{The preferred departure time of vehicles, the time window of vehicle 3 and its actions. 
}  
\label{fig:timewind}
   \end{figure}
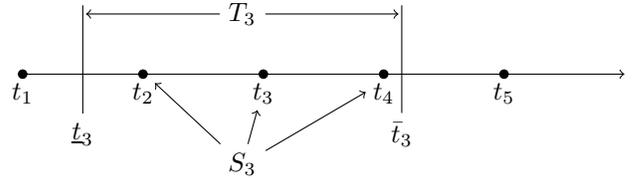

Next, we introduce three sets which are helpful in the game-theoretic formulation of the platoon matching problem. The set $C_j(s)=\{i | s_i=t_j, i\in\mathcal{N} \}$ denotes the vehicles which select their departure times as $t_j$. Thus, if $|C_j(s)|$ is more than one, the vehicles in $C_j(s)$ will at least platoon over the outgoing edge of $v_1$. Let $C(s)=\{C_j (s) | t_j \in \mathcal{T} \}$ represent the set of all potential platoons depending on the actions of the vehicles.

\subsection{Utility of Vehicles} \label{utilitysec}

In this section, we first present a model of the monetary saving from platooning on one road segment. Then, the utility functions of vehicles are defined.

\subsubsection{Model of Monetary Saving From Platooning}\label{fuelsave}

The leading vehicle in a platoon is commonly called the platoon leader and the other vehicles are commonly called the platoon followers. The field experiment \cite{Browand2004} evaluated the fuel reduction of two similar vehicles due to platooning. The experiment took place on a relatively flat road. When the inter-vehicle distance was $3-10$ m, the total reduction of fuel consumption was $10-12\%$ for the platoon follower and $5-10\%$ for the platoon leader. The field experiment \cite{Alam2010} showed a reduction of fuel consumption between $1-7 \%$  for the platoon follower depending on the gap between the vehicles.   

According to the above experiments, the aggregated fuel saving in a platoon depends on the platooning distance and the number of vehicles in the platoon. Thus, we assume that on each road segment, the aggregated monetary saving from platooning is a function of the length of the platoon and the length of the road segment. Consider a platoon with vehicles $1,\dots,N$, where the platoon members are ordered according to their positions in the platoon, \emph{i.e.,} vehicle $1$ is the platoon leader and so on. Let $\tilde f_N(i)$ represent the monetary saving, due to fuel saving per meter of the $i$-th vehicle in the platoon. The aggregated monetary saving of a platoon of length $N$ over a road segment of length $d$ is then given by 

\begin{equation}
\sum \limits_{i=1}^N  \tilde f_N(i) d.  
\end{equation}
Moreover, we assume that the aggregated monetary saving from platooning is equally shared among the vehicles in the platoon. Hence, after sharing, the monetary saving per meter of each vehicle in the platoon is   

\begin{equation}
f(N)=\frac{1}{N}\sum \limits_{i=1}^N  \tilde f_N(i).
\end{equation}

\begin{remark}

 In this paper, the game-theoretic analysis is provided for a general function $f(\cdot)$, but for the numerical results, we assume $f(N)=k_p\frac{N-1}{N}$. This special model for $f(\cdot)$ is suitable when the fuel saving of the platoon leader is zero, \emph{i.e.,} $ \tilde f_N(1)=0$, and the fuel saving for each platoon follower is $k_p$ per meter, \emph{i.e.,} $ \tilde f_N(i)=k_p$ for $i=2,...,N$. 

\end{remark}

\subsubsection{Utility Functions}

In this section, we define the vehicles' utility functions. The model allows the vehicles to have different destinations so the platoons may split into sub-platoons along the route, e.g., see Fig. \ref{fig:style}. In this figure, the vehicles $1,2,3$ have destinations $v_5,v_3,v_4$, respectively. The vehicles have the same departure time from $v_1$ and form a platoon on the first road segment. The platoon splits at $v_2$ and the vehicles $1$ and $2$ form a sub-platoon on the road segment between $v_2$ and $v_3$. 
The model captures the monetary loss of a vehicle if it deviates from its preferred departure time, e.g., due to higher probability of missing a dead-line, longer working day for the driver and penalty for cutting a rest for the driver. We start by defining the utility function for a general function $f(\cdot)$.

\begin{model}[A Non-Cooperative Utility] \label{modelu1}
For each vehicle $i \in \mathcal N$ and given actions $s_i=t_j$ and $s_{-i}$, i.e., vehicle $i$ is a part of a platoon consisting of the vehicles $C_j(s)$, the utility of vehicle $i$ is defined as

\begin{equation}\label{utility}
u_i(s)= \sum \limits_{e \in P(i)}\!\!\! \big(f( n(e, C_j)) d(e) \big) -\beta(t_j,t_i), 
\end{equation}
where the sum represents the monetary saving due to platooning. The term $\beta(t_j,t_i)$ represents the loss for adapting the departure time from $t_i$ to $t_j$. 
\end{model}

In our numerical results in Section \ref{simulation}, Model \ref{modelu1} is used with $f(n)=k_p\frac{n-1}{n}$ and  $\beta(t_j,t_i)=k_t|t_j-t_i|$. This model is accurate when on each road segment the leader has zero fuel saving, the followers have equal fuel saving and the total monetary saving from platooning is shared among the platoon members.

 \begin{remark}
 
  When a platoon is formed, we assume that one of the platoon members is randomly selected as the leader.  
\end{remark}

\begin{remark}

Although we only considered monetary saving due to the fuel saving above, there are other benefits from platooning which can be considered in the same framework, e.g., reduction in $\text{CO}_2$ emissions and workload of drivers, safety improvements and reduction of traffic congestion. 

\end{remark}

\begin{figure}
	\centering
	\includegraphics[scale=0.4]{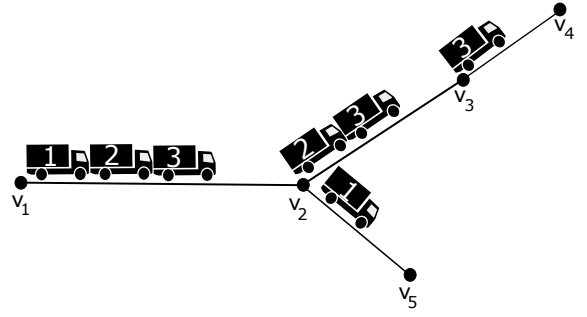}
	\caption{A platoon of three vehicles that splits along the route.}
	\label{fig:style}
\end{figure}

\subsection{Platoon Matching Game}\label{gametheory}

 The interaction among the vehicles is modeled by a non-cooperative game. The players are the vehicles in $\mathcal{N}$. The action of each vehicle $i\in \mathcal N$ is its departure time $s_i~\in~S_i$ and the strategy space of the game is $S=S_1 \times \cdots \times  S_{ N}$, where $\times$ denotes the Cartesian product.   The platoon matching game is formally defined by the triplet $\mathcal G~=~(\mathcal N, S, U)$, where  $U(s)=\{u_i(s)|i\in \mathcal N \}$.

\subsection{Cooperative Solution}

In the numerical results, we also consider an aggregated utility function which can be used to obtain a cooperative solution for the platooning problem. In the cooperative scenario, vehicles have a common utility function which is defined next.

\begin{model}[A Cooperative Utility]\label{utilityagg}
	In the cooperative scenario the common utility function of all vehicles is defined as 
	
	\begin{equation}
	u(s)=\sum \limits_{k \in \mathcal N} \bigg ( \sum \limits_{e \in P(k)}\!\!\! \Big(f( n(e, C_{I_k}(s))) d(e) \Big) -\beta(s_k,t_k)  \bigg ),
	\end{equation}
	where $s_k$ is the decision of vehicle $k$ and $I_k(s_k)=j$ if $s_k=t_j$.
\end{model}

\section{Equilibrium Analysis of the Platoon Matching Game} \label{DG}

We consider pure Nash equilibrium (NE) as the solution concept for the game $\mathcal G$. We first study the existence of NE for the platoon matching game. Then, an algorithm is proposed for finding an NE. We start by defining the NE. A pure NE is a strategy profile $s^* \in S$ such that

\begin{equation}
u_i(s_i^*,s_{-i}^*)\geq u_i(s_i,s_{-i}^*),\ \forall s_i\in S_i, \ \forall i\in \mathcal N. 
\end{equation}

\subsection{Existence of NE}\label{defgamre}

Before stating the main result of this subsection, we define the notion of a potential game. The game $\mathcal G$ is called an exact potential game if there exists a function $\Phi: S \rightarrow \mathbb{R}$ such that, for all $i\in \mathcal{N}$ and all $s'_{i}, s''_{i}\in S_i, s_{-i}\in S_{-i}$ , we have

\begin{equation}
  \Phi(s_i',s_{-i})-\Phi(s_i'',s_{-i})=u_i(s_{i}',s_{-i})-u_i(s_{i}'',s_{-i}),
\end{equation}
where $S_{-j}=S_1 \times \cdots \times  S_{i-1} \times  S_{i+1} \times    \cdots\times   S_{N}$. It is shown in \cite{Monderer1996} that finite exact potential games have at least one NE. The next lemma shows that $\mathcal G$ is an exact potential game.

\begin{theorem}\label{lemma1}

Consider the platoon matching game $\mathcal G$, with the utility functions given by Model \ref{modelu1}. Let $f(\cdot)$ be an arbitrary mapping from $\mathbb N$ to $[0,f_{max}]$ with  $0\leq f_{max}<\infty$. Then, $\mathcal G$ is a potential game with the potential function 

    \begin{equation}\label{potent}
\Phi(s) = \sum  \limits_{C \in C(s )} \!\!\! \phi(C) -   \sum \limits_{l \in \mathcal{N} } \beta(t_l,s_l ),
\end{equation}
where
    
\begin{equation}
    \phi(C)= \sum \limits_{e \in P (C)} r(n(e,C)) d(e)
\end{equation}    
and    
    
  \begin{equation}
 r(n)=\sum \limits_{i=1}^{n}f(n).
\end{equation}

\end{theorem}

\begin{proof}
See Appendix \ref{appen}.
\end{proof}

 \subsection{Equilibrium Seeking Algorithm}\label{algsec}

 The authors in \cite{Monderer1996} showed that in finite exact potential games, the best response dynamics converges to an NE. In the best response dynamics, the players update their actions according to their best response functions, one at a time. The best response function of vehicle $i$ to $s_{-i}$ is defined as 

\begin{equation}
B_i (s_{-i})=  \underset{s_i\in S_i }{\arg\max}   \ u_i(s_i, s_{-i}).
\end{equation}


 Algorithm \ref{alg} presents a platoon matching algorithm based on the best response dynamics. If two vehicles have the same equilibrium action then they will be in the same platoon. This algorithm is used in the numerical results. 

\begin{algorithm} 
\SetKwInOut{Input}{input}\SetKwInOut{Output}{output}
\Input{Initial strategy profile, $s^0=[t_1,...,t_N]^T$}
\Output{NE, $s^*  $}
\BlankLine
$q=0$ \\
$\hat s^{q-1}=s^0$\\
\While{$s^{q}   \neq \hat s^{q-1} \lor q=0 $ }{ $\hat s^{q}  =s^q $ \\ 
\For{$i \in \mathcal{N}$}{
  $s_i^{q}   =B_i (s^q_{-i} )$
 }
$q=q+1$ \\
$s^q=s^{q-1}$\\}
 
$s^* =s^q $
  \caption{Nash equilibrium seeking algorithm}
 \label{alg}
\end{algorithm}

\section{Numerical Results}\label{simulation}

In this section, Algorithm \ref{alg} is used to compute the equilibrium of the platoon matching game among a set of vehicles. First, the set-up of the numerical experiments is explained. Second, the convergence of the algorithm is demonstrated with an example. Then, we study the fuel saving and the proportion of non-platooning vehicles when the preferred departure times are random.

\subsection{Simulation Set-up}

 We consider the road network shown in Fig. \ref{fig:RN}, where the vehicles are located at node $v_1$. Vehicle  destinations are randomly drawn from $v_2,...,v_{13}$ according to a uniform distribution. The preferred departure time of each vehicle $i$ is drawn according to a uniform distribution on the time interval $[0,\alpha]$, \emph{i.e.,} $t_i \sim U(0,\alpha)$. Note that $\alpha$ affects the variance of the preferred departure times of the vehicles. Thus, a large value of $\alpha$ implies that the preferred departure times are scattered over a wide interval. The time window of each vehicle $i$ is defined as $T_i=[t_i-500, t_i+~500 ]$, \emph{i.e.,} the length of the time window is $1000$ seconds and it is centered at the preferred departure time $t_i$.  The individual utility function of each vehicle is defined by 
 Model \ref{modelu1}, with $f(n)=k_p\frac{n-1}{n}$, $\beta(t_j,t_i)=k_t|t_j-t_i|$, $k_t=1.5\times 10^{-2}$ and $k_p=5\times 10^{-5}$. We assume that each liter of fuel costs one dollar, thus, each platoon follower saves $k_p=5\times 10^{-5}$ liters of fuel per followed meter.

\subsection{Convergence of Algorithm \ref{alg} }

In this sub-section, we first study the evolution of the strategy profile of vehicles under Algorithm \ref{alg}. For the numerical results in this sub-section the set of vehicles is $\mathcal N=\{1,...,5 \}$ and $\alpha$ is set to $15000$. 

The convergence of the vehicles' actions under Algorithm~\ref{alg} is illustrated in Fig. \ref{fig:match}. The initial action of each vehicle is its preferred departure time. The evolution of the strategy profile is shown in Fig. \ref{fig:match}. According to this figure, the strategy profile converge to an NE. Vehicles $1,3,4,5$ reach an agreement and are therefore matched. Vehicle $2$ has a unique equilibrium action and departs on its own preferred departure time.

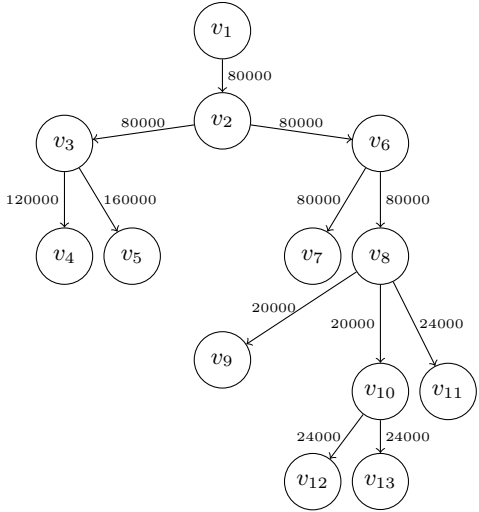
\begin{figure}
	\small
	\centering
	\begin{tikzpicture}[scale=0.6]
	\path   
	(0,0) node(x1)[circle,draw=black, fill=white,align=left,below, minimum size=0.75 cm] {$v_1$} 
	(0,-2) node(x2)[circle,draw=black, fill=white,align=left,   below, minimum size=0.75 cm] {$v_2$} 
	(-3.5,-2.5) node(x3)[circle,draw=black, fill=white,align=left,   below, minimum size=0.75 cm]  {$v_3$}
	(-3.5,-5) node(x4)[circle,draw=black, fill=white,align=left,   below, minimum size=0.75 cm]  {$v_4$}
	(-2,-5) node(x5)[circle,draw=black, fill=white,align=left,   below, minimum size=0.75 cm]  {$v_5$}
	(3.5,-2.5) node(x6)[circle,draw=black, fill=white,align=left,   below, minimum size=0.75 cm]  {$v_6$}
	(2,-5) node(x7)[circle,draw=black, fill=white,align=left,   below, minimum size=0.75 cm]  {$v_7$}
	(3.5,-5) node(x8)[circle,draw=black, fill=white,align=left,   below, minimum size=0.75 cm]  {$v_8$}
	(0,-7.3) node(x9)[circle,draw=black, fill=white,align=left,   below, minimum size=0.75 cm]  {$v_9$}
	(3.5,-8) node(x10)[circle,draw=black, fill=white,align=left,   below, minimum size=0.75 cm]  {$v_{10}$}
	(5,-8) node(x11)[circle,draw=black, fill=white,align=left,   below, minimum size=0.75 cm]  {$v_{11}$}
	(2,-10) node(x12)[circle,draw=black, fill=white,align=left,   below, minimum size=0.75 cm]  {$v_{12}$}
		(3.5,-10) node(x13)[circle,draw=black, fill=white,align=left,   below, minimum size=0.65 cm]  {$v_{13}$}
	 ;
	 \tiny
	\draw [->] (x1) -- (x2)node [midway,right ]  {$80000$}; 
	\draw [->] (x2) -- (x3)node [midway,above]  {$80000$}; 
	\draw [->] (x3) -- (x4)node [midway,right,left ]  {$120000$}; 
	\draw [->] (x3) -- (x5)node [midway,right,right ]  {$160000$}; 
	\draw [->] (x2) -- (x6)node [midway,,above ]  {$80000$}; 
	\draw [->] (x6) -- (x7)node [midway,right,left ]  {$80000$};
	\draw [->] (x6) -- (x8)node [midway,right,right ]  {$80000$};  
	\draw [->] (x8) -- (x9)node [midway,right,left ]  {$20000$};
	\draw [->] (x8) -- (x10)node [midway,right,left ]  {$20000$};  
	\draw [->] (x8) -- (x11)node [midway,right,right ]  {$24000$};  
	\draw [->] (x10) -- (x12)node [midway,right,left ]  {$24000$}; 
	\draw [->] (x10) -- (x13)node [midway,right,right ]  {$24000$}; 
	\end{tikzpicture}
	
	\normalsize
	\caption{The road network used for numerical results. 
	} \label{fig:RN}
\end{figure}

\begin{figure}
    \centering
    \includegraphics[scale=0.35]{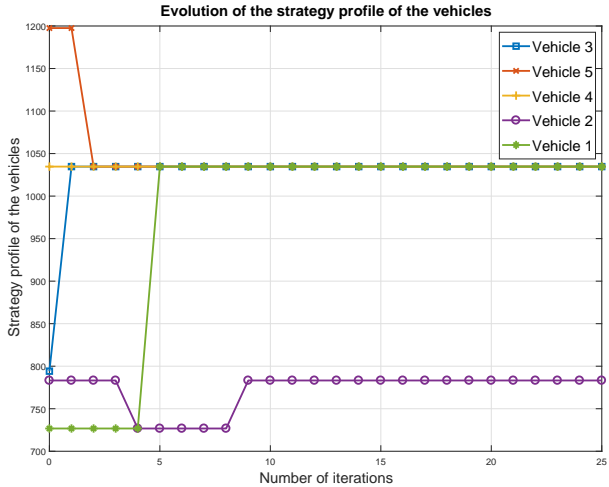}
    \caption{The evolution of the strategy profile of the vehicles under the Algorithm~\ref{alg}.}
    \label{fig:match}
\end{figure}

\subsection{Impact of $\alpha$ on the Game Outcome}\label{simset}

In this sub-section, we study the impact of $\alpha$ on the total fuel saving and the proportion of vehicles which leave the origin node without platooning. We also compare the results with a cooperative solution which is obtained by sequentially maximizing  the sum of the vehicle utility functions. At  each time, the action of one vehicle is adjusted such that the sum of the utilities is maximized while the actions of other vehicles are fixed. This approach is repeated until it converges. The limiting actions are a cooperative solution to the platoon matching problem.

For the numerical results in this sub-section, we consider $10$ vehicles. The variance of the preferred departure times is varied, \emph{i.e.,} $\alpha$ is varied and we study its impact on the total fuel saving of the vehicles and the proportion of non-platooning vehicles. For each value of $\alpha$, we repeated the numerical experiment $100$ times.

 In Fig. \ref{case1}, the total fuel saving and the average proportion of non-platooning vehicles are shown when $\alpha$ is varied from $0$ to $1500$. According to this figure, the total fuel saving decreases as $\alpha$ becomes large, while the proportion of non-platooning vehicles increases. Note that when the preferred departure times of the vehicles are close, the penalty terms, due to the difference between the preferred and actual departure times, are relatively small and more vehicles will therefore platoon. 

Moreover, the cooperative solution, in general, results in a higher total fuel saving and smaller proportion of non-platooning vehicles, compared with the NE solution. Note that the individual utility for platooning is higher for the cooperative case, which incentivizes the vehicles to platoon. 
 However, for small $\alpha$, the equilibrium solution resulted in a higher fuel saving.

\begin{figure}
     \centering
     \subfloat[][Average total fuel saving as a function of $\alpha$.]{\includegraphics[scale=0.50]{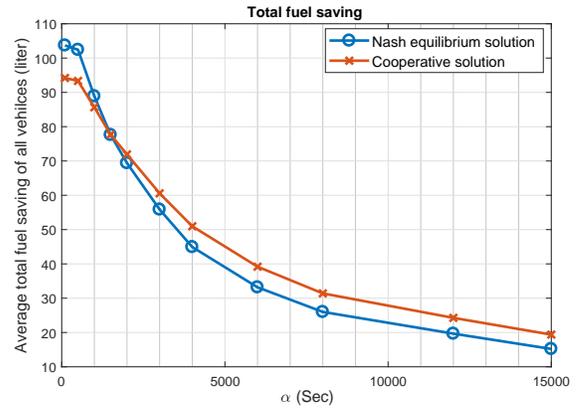}\label{fig:Ut}}
     \\\subfloat[][Average proportion of non-platooning vehicles as a function of $\alpha$ ]{\includegraphics[scale=0.50]{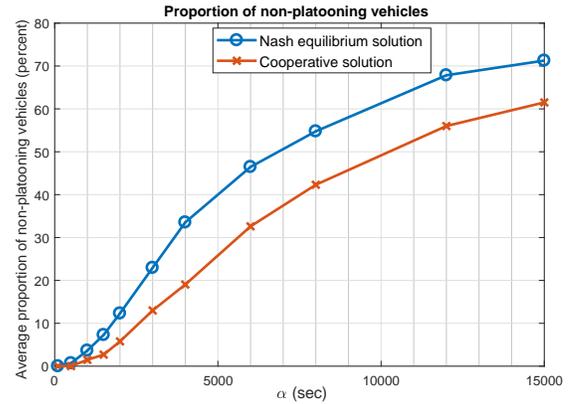}\label{fig:HTA}}
     \caption{Average total fuel saving as a function of $\alpha$ (a) and average proportion of non-platooning vehicles as a function of $\alpha$ (b). }
     \label{case1}
\end{figure}

\section{Conclusions and future work}\label{conclusions}

We proposed a non-cooperative game which models the multi-fleet platoon matching problem. It was shown that the game has at least one pure  NE. An algorithm based on the best response dynamics was proposed for finding an NE. The algorithm was used to study the impact of the variance of the preferred departure times on the total fuel saving and proportion of non-platooning vehicles. The total fuel saving and the proportion of non-platooning vehicles highly depend on the variance of the preferred departure times. Moreover, the non-cooperative solution was compared with a cooperative solution. The numerical results shows that the cooperative solution results in a larger fuel saving and a smaller proportion of non-platooning vehicles compared with the non-cooperative solution.

In the future we will study the platoon matching problem when the vehicles have different origins and destinations in a road network defined on a graph with general topology and each vehicle stops at several nodes where it can be matched with other vehicles. Furthermore, we plan to take into account uncertain traveling times on the road segments.

\appendices
\section{Proof of Theorem \ref{lemma1}}\label{appen}
We aim to show that $\Phi (s_i,s_{-i})-\Phi (s_i',s_{-i})$ equals  $u_i(s_i,s_{-i})-u_i(s_i',s_{-i})$, where $s_i$ and $s_i'$ are two different actions of an arbitrary vehicle $i$. We denote the terms which are not affected by alternating the action of vehicle $i$ from $s_i$ to $s_i'$ by
\begin{align*}
     O(s)= &  \sum \limits_{C \in C(s)  \setminus  \{C_j(s),C_{j'}(s)\}}   \phi(C) +\\   &   \sum \limits_{e \in P (C_j(s)) \setminus P(i) }     r\big(n(e,C_{j}(s))\big) d(e)+ \\   &
      \sum \limits_{e \in P (C_{j'}(s))  \setminus  P(i) }     r\big(n(e,C_{j'}(s))\big) d(e)    - \\  &  \sum \limits_{l \in \mathcal{N}  \setminus  \{i\}} \beta(t_l,s_l). 
\end{align*}
We then have

\begin{align*}
   \Phi(s_i,s_{-i})= &  \sum\limits_{e \in P(i)} r\big(n(e,C_j(s))\big) d(e)+ \\ &   \sum\limits_{e \in P(i) } r\big(n(e,C_{j'}(s))\big) d(e)- \\ &  \beta(t_i,s_i)+ O(s)
\end{align*}
Note that if $s_i=t_j$ then $n(e,C_{j}(s))=n(e,C_{j}(s)/\{i\})+1$ and if $s_i\neq t_j$ then $n(e,C_{j}(s))=n(e,C_{j}(s)/\{i\})$. Hence, if $s_i=t_j$ we have

\begin{align*}
    \Phi(s_i,s_{-i})=  &  \sum\limits_{e \in P(i)} r\big(n(e,C_j(s)/\{i\})+1 \big) d(e)+ \\ &  \sum\limits_{e \in P(i)} r \big(n(e,C_{j'}(s)/\{i\}) \big) d(e)  -  \\ &  \beta(t_i,s_i) + O(s), 
\end{align*}
and similar when $s_i'=t_{j'}$. We define $\Delta \Phi =\Phi(s_i,s_{-i})-\Phi(s_i',s_{-i})$. Then we have
\small
\begin{align*}
  \Delta \Phi= & \!   \sum\limits_{e \in P(i)} \!\!\!\! \bigg(r \Big(n(C_j(s)/\{i\},e)+ \! \! 1 \!  \Big) \! - r  \Big(n(e,C_j(s)/\{i\}) \Big)  \! \! \bigg   )d(e)  -  \\ &\beta(t_i,s_{i}) -  \\ & \!   \sum\limits_{e \in P(i)}  \!\!\!\!   \bigg(r \Big(n(e,C_{j'}(s)/\{i\}) \!    +1 \!   \Big) \! \! -   \!    r\Big( \! n(e,C_{j'}(s)/\{i\})\Big  )  \! \!\bigg)d(e) + \\ &   \beta(t_i,s_{i}').
\end{align*}
\normalsize
We define $\Delta u_i=u_i(s_i,s_{-i})-u_i(s_i',s_{-i})$. Then we have

\begin{align*}
  \Delta u_i= &   \sum\limits_{e \in P(i)} \!\!\!\! f \Big(n(C_j(s)/\{i\},e)+1 )  \bigg)d(e) - \\ &\beta(t_i,s_{i}) -  \\ &   \sum\limits_{e \in P(i)}  \!\!\!\! f \Big(n(e,C_{j'}(s)/\{i\})+1 ) \Big)d(e)  + \\ &  \beta(t_i,s_{i}').
\end{align*}
Hence, if $r(\cdot)$ satisfies

\begin{equation}
 r(n+1)-r(n)=f(n+1), \ \ n=0,1,...      
\end{equation}
then we have $\Delta \Phi= \Delta u_i$. But such $r(\cdot)$ does always exists since

\begin{equation}
 r(n)-r(0)=\sum \limits_{i=1}^{n}r(i)-r(i-1) =\sum \limits_{i=1}^{n}f(n),     
\end{equation}
by $r(0)=0$. The conclusions follows.

\qed

\bibliographystyle{ieeetr} %
\bibliography{RefDatabase}

\end{document}